\documentclass[submission,copyright,creativecommons]{eptcs}
\providecommand{\event}{Linearity and TLLA 2018}
\usepackage{url,doi}
\usepackage[british]{babel}
\usepackage[all]{foreign}
\usepackage{amsmath,amsthm,amssymb}
\allowdisplaybreaks
\def\Coloneqq{\mathbin{::=}}
\def\coloneqq{\mathbin{:=}}
\usepackage[inline]{enumitem}%
\usepackage{breakurl} % Not needed if you use pdflatex only.
\usepackage{textcomp,textgreek,upgreek,stmaryrd}
\usepackage{nameref,hyperref}
\usepackage[dvipsnames,table]{xcolor}
\usepackage{scalerel}
\usepackage{rotating}
\usepackage{centernot}
\usepackage{xspace}
\xspaceaddexceptions{]\}}
\usepackage[nameinlink,noabbrev,capitalize]{cleveref}
\hypersetup{hidelinks,final}
\newtheorem{lemma}{Lemma}[section]
\newtheorem{theorem}[lemma]{Theorem}
\newtheorem{corollary}[lemma]{Corollary}
\newtheorem{definition}[lemma]{Definition}

\definecolor{tmcolor}{HTML}{C02F1D}
\definecolor{tycolor}{HTML}{1496BB}
\definecolor{cscolor}{HTML}{A3B86C}
\providecommand{\tm}[1]{\textcolor{Red}{\ensuremath{\normalfont#1}}}
\providecommand{\ty}[1]{\textcolor{NavyBlue}{\ensuremath{\normalfont#1}}}

\usepackage{bussproofs}
\EnableBpAbbreviations

\providecommand{\seq}[2][]{\ensuremath{\tm{#1}\;\vdash\;\ty{#2}}}

\providecommand{\tmty}[2]{\ensuremath{\tm{#1}\colon\!\ty{#2}}}
\providecommand{\NOM}[1]{\RightLabel{\textsc{#1}}}
\providecommand{\SYM}[1]{\RightLabel{\ensuremath{#1}}}

\newenvironment{prooftree*}{\leavevmode\hbox\bgroup}{\DisplayProof\egroup}
  {\gdef\scalefactor{#1}\begin{center}\proofSkipAmount \leavevmode}%
  {\scalebox{\scalefactor}{\DisplayProof}\proofSkipAmount \end{center} }
\newenvironment{scprooftree*}[1][1]%
  {\gdef\scalefactor{#1}\leavevmode\hbox\bgroup}
  {\DisplayProof\egroup} 
\providecommand{\moveMixDown}{\leadsto}

\providecommand{\mtf}[1]{\ensuremath{\llbracket #1 \rrbracket}}

\providecommand{\fv}[1]{\ensuremath{\text{fv}(#1)}}

\providecommand{\notFreeIn}[2]{\tm{#1}\not\in\tm{#2}}

\providecommand{\reducesto}[3][]{\ensuremath{\tm{#2}\overset{#1}{\Longrightarrow}\tm{#3}}}

\providecommand{\ppar}{\ensuremath{\mid}}

\providecommand{\piBoundSend}[3]{\ensuremath{#1[ #2 ].#3}}
\providecommand{\piRecv}[3]{\ensuremath{#1( #2 ).#3}}
\providecommand{\piPar}[2]{\ensuremath{#1 \ppar #2}}
\providecommand{\piNew}[2]{\ensuremath{(\nu #1)#2}}

\providecommand{\piHalt}[0]{\ensuremath{0}}
\providecommand{\piSub}[3]{\ensuremath{#3\{#1/#2\}}}

\providecommand{\piDILL}{\textpi DILL\xspace}
\providecommand{\cp}{CP\xspace}

\providecommand{\cpLink}[2]{\ensuremath{#1{\leftrightarrow}#2}}
\providecommand{\cpCut}[3]{\ensuremath{\piNew{#1}({\piPar{#2}{#3}})}}
\providecommand{\cpSend}[4]{\ensuremath{#1[#2].(\piPar{#3}{#4})}}
\providecommand{\cpRecv}[3]{\ensuremath{#1(#2).#3}}
\providecommand{\cpWait}[2]{\ensuremath{#1().#2}}
\providecommand{\cpHalt}[1]{\ensuremath{#1[].0}}
\providecommand{\cpInl}[2]{\ensuremath{#1\triangleleft\texttt{inl}.#2}}
\providecommand{\cpInr}[2]{\ensuremath{#1\triangleleft\texttt{inr}.#2}}
\providecommand{\cpCase}[3]{\ensuremath{#1\triangleright\{\texttt{inl}:#2;\texttt{inr}:#3\}}}
\providecommand{\cpAbsurd}[1]{\ensuremath{#1\triangleright\{\}}}
\providecommand{\cpSub}[3]{\ensuremath{\piSub{#1}{#2}{#3}}}

\providecommand{\parr}{\mathbin{\bindnasrepma}}
\providecommand{\with}{\mathbin{\binampersand}}
\providecommand{\plus}{\ensuremath{\oplus}}
\providecommand{\tens}{\ensuremath{\otimes}}
\providecommand{\one}{\ensuremath{\mathbf{1}}}
\providecommand{\nil}{\ensuremath{\mathbf{0}}}

\providecommand{\emptycontext}{\ensuremath{\,\cdot\,}}
\providecommand{\bigtens}{\ensuremath{\scalerel*{\tens}{\sum}}}
\providecommand{\bigparr}{\ensuremath{\scalerel*{\parr}{\sum}}}

\newcommand{\cpEquivLinkComm}{\ensuremath{(\cpLink{}{}\text{-sym})}\xspace}
\newcommand{\cpEquivCutComm}{\ensuremath{(\nu\text{-comm})}\xspace}
\newcommand{\cpEquivCutAssNoParen}[1]{\ensuremath{\nu\text{-assoc}\xspace}}
\newcommand{\cpEquivCutAss}[1]{\ensuremath{(\cpEquivCutAssNoParen{#1})}\xspace}
\newcommand{\cpRedAxCut}[1]{\ensuremath{(\cpLink{}{})}\xspace}
\newcommand{\cpRedBetaTensParr}{\ensuremath{(\beta{\tens}{\parr})}\xspace}
\newcommand{\cpRedBetaOneBot}{\ensuremath{(\beta{\one}{\bot})}\xspace}
\newcommand{\cpRedBetaPlusWith}[1]{\ensuremath{(\beta{\plus}{\with}_{#1})}\xspace}

\newcommand{\cpRedGammaCut}{\ensuremath{(\gamma\nu)}\xspace}
\newcommand{\cpRedGammaEquiv}{\ensuremath{(\gamma{\equiv})}\xspace}

\providecommand{\cpInfAx}{%
  \begin{prooftree*}
    \AXC{$\vphantom{\seq[ Q ]{ \Delta, \tmty{y}{A^\bot} }}$}
    \NOM{Ax}
    \UIC{$\seq[ \cpLink{x}{y} ]{ \tmty{x}{A}, \tmty{y}{A^\bot} }$}
  \end{prooftree*}}
\providecommand{\cpInfCut}{%
  \begin{prooftree*}
    \AXC{$\seq[ P ]{ \Gamma, \tmty{x}{A} }$}
    \AXC{$\seq[ Q ]{ \Delta, \tmty{x}{A^\bot} }$}
    \NOM{Cut}
    \BIC{$\seq[ \cpCut{x}{P}{Q} ]{ \Gamma, \Delta }$}
  \end{prooftree*}}

\providecommand{\cpInfTens}{%
  \begin{prooftree*}
    \AXC{$\seq[ P ]{ \Gamma , \tmty{y}{A} }$}
    \AXC{$\seq[ Q ]{ \Delta , \tmty{x}{B} }$}
    \SYM{(\tens)}
    \BIC{$\seq[ \cpSend{x}{y}{P}{Q} ]{ \Gamma , \Delta , \tmty{x}{A \tens B} }$}
  \end{prooftree*}}
\providecommand{\cpInfParr}{%
  \begin{prooftree*}
    \AXC{$\seq[ P ]{ \Gamma , \tmty{y}{A} , \tmty{x}{B} }$}
    \SYM{(\parr)}
    \UIC{$\seq[ \cpRecv{x}{y}{P} ]{ \Gamma , \tmty{x}{A \parr B} }$}
  \end{prooftree*}}
\providecommand{\cpInfOne}{%
  \begin{prooftree*}
    \AXC{$\vphantom{\seq[ P ]{ \Gamma }}$}
    \SYM{(\one)}
    \UIC{$\seq[ \cpHalt{x} ]{ \tmty{x}{\one} }$}
  \end{prooftree*}}
\providecommand{\cpInfBot}{%
  \begin{prooftree*}
    \AXC{$\seq[ P ]{ \Gamma }$}
    \SYM{(\bot)}
    \UIC{$\seq[ \cpWait{x}{P} ]{ \Gamma , \tmty{x}{\bot} }$}
  \end{prooftree*}}
\providecommand{\cpInfPlus}[1]{%
  \ifdim#1pt=1pt
  \begin{prooftree*}
    \AXC{$\seq[ P ]{ \Gamma , \tmty{x}{A} }$}
    \SYM{(\plus_1)}
    \UIC{$\seq[{ \cpInl{x}{P} }]{ \Gamma , \tmty{x}{A \plus B} }$}
  \end{prooftree*}
  \else%
  \ifdim#1pt=2pt
  \begin{prooftree*}
    \AXC{$\seq[ P ]{ \Gamma , \tmty{x}{B} }$}
    \SYM{(\plus_2)}
    \UIC{$\seq[ \cpInr{x}{P} ]{ \Gamma , \tmty{x}{A \plus B} }$}
  \end{prooftree*}
  \else%
  \fi%
  \fi%
}
\providecommand{\cpInfWith}{%
  \begin{prooftree*}
    \AXC{$\seq[ P ]{ \Gamma , \tmty{x}{A} }$}
    \AXC{$\seq[ Q ]{ \Gamma , \tmty{x}{B} }$}
    \SYM{(\with)}
    \BIC{$\seq[ \cpCase{x}{P}{Q} ]{ \Gamma , \tmty{x}{A \with B} }$}
  \end{prooftree*}}
\providecommand{\cpInfNil}{%
  \text{(no rule for \ty{\nil})}}
\providecommand{\cpInfTop}{%
  \begin{prooftree*}
    \AXC{}
    \SYM{(\top)}
    \UIC{$\seq[ \cpAbsurd{x} ]{ \Gamma, \tmty{x}{\top} }$}
  \end{prooftree*}}

\providecommand{\dhcp}{HCP\xspace}
\providecommand{\hcp}{HCP$^{{-}}$\xspace}
\providecommand{\hsep}{\ensuremath{\mid}}
\providecommand{\emptyhypercontext}{\varnothing}

\providecommand{\hcpEquivMixComm}{\ensuremath{({\ppar}\text{-comm})}\xspace}
\providecommand{\hcpEquivMixAssNoParen}[1]{\ensuremath{{\ppar}\text{-assoc}}\xspace}
\providecommand{\hcpEquivMixAss}[1]{\ensuremath{(\hcpEquivMixAssNoParen{#1})}\xspace}
\providecommand{\hcpEquivMixHaltNoParen}[1]{\ensuremath{\text{halt}\xspace}}
\providecommand{\hcpEquivMixHalt}[1]{\ensuremath{(\hcpEquivMixHaltNoParen{#1})}\xspace}
\providecommand{\hcpEquivMixCut}[1]{\ensuremath{(\text{scope-ext})}\xspace}

\providecommand{\hcp}{HCP\xspace}

\providecommand{\hcpEquivMixComm}{%
  \ensuremath{({\ppar}\text{-comm})}\xspace}
\providecommand{\hcpEquivMixAssNoParen}[1]{%
  \ensuremath{{\ppar}\text{-assoc}}\xspace}
\providecommand{\hcpEquivMixAss}[1]{%
  \ensuremath{(\hcpEquivMixAssNoParen{#1})}\xspace}
\providecommand{\hcpEquivScopeExtNoParen}[1]{%
  \ensuremath{\text{scope-ext}\xspace}}
\providecommand{\hcpEquivScopeExt}[1]{%
  \ensuremath{(\hcpEquivScopeExtNoParen{#1})}\xspace}
\providecommand{\hcpEquivMixHaltNoParen}[1]{%
  \hcpEquivMixHaltNoParen{#1}}
\providecommand{\hcpEquivMixHalt}[1]{%
  \hcpEquivMixHalt{#1}}
\providecommand{\hcpEquivNewComm}{%
  \ensuremath{(\nu\text{-comm})}\xspace}

\providecommand{\hcpRedAxCut}[1]{%
  \ensuremath{(\cpLink{}{})}\xspace}
\providecommand{\hcpRedBetaTensParr}{%
  \ensuremath{(\beta{\tens}{\parr})}\xspace}
\providecommand{\hcpRedBetaOneBot}{%
  \ensuremath{(\beta{\one}{\bot})}\xspace}
\providecommand{\hcpRedBetaPlusWith}[1]{%
  \ensuremath{(\beta{\plus}{\with}_{#1})}\xspace}
\providecommand{\hcpRedGammaNew}{%
  \ensuremath{(\gamma{\nu})}\xspace}
\providecommand{\hcpRedGammaMix}{%
  \ensuremath{(\gamma{\ppar})}\xspace}
\providecommand{\hcpRedGammaEquiv}{%
  \ensuremath{(\gamma{\equiv})}\xspace}

\providecommand{\hcpInfAx}{\cpInfAx}
\providecommand{\hcpInfCut}{%
  \begin{prooftree*}
    \AXC{$\seq[{ P }]{
        \mathcal{G}\hsep
        \Gamma, \tmty{x}{A} \hsep
        \Delta, \tmty{x}{A^\bot}
      }$}
    \NOM{Cut}
    \UIC{$\seq[{ \piNew{x}{P} }]{
        \mathcal{G}\hsep
        \ty{\Gamma}, \ty{\Delta}
      }$}
  \end{prooftree*}}
\providecommand{\hcpInfMix}{%
  \begin{prooftree*}
    \AXC{$\seq[ P ]{\mathcal{G} }$}
    \AXC{$\seq[ Q ]{\mathcal{H} }$}
    \NOM{H-Mix}
    \BIC{$\seq[ \piPar{P}{Q} ]{
        \mathcal{G} \hsep \mathcal{H} }$}
  \end{prooftree*}}
\providecommand{\hcpInfHalt}{%
  \begin{prooftree*}
    \AXC{$\vphantom{\seq[ Q ]{ \Delta, \tmty{y}{A^\bot} }}$}
    \NOM{H-Mix$_0$}
    \UIC{$\seq[{ \piHalt }]{ \emptyhypercontext }$}
  \end{prooftree*}}

\providecommand{\hcpInfBoundTens}{%
  \begin{prooftree*}
    \AXC{$\seq[{ P }]{
        \mathcal{G} \hsep
        \ty{\Gamma}, \tmty{y}{A} \hsep
        \ty{\Delta}, \tmty{x}{B}
      }$}
    \SYM{\tens}
    \UIC{$\seq[{ \piBoundSend{x}{y}{P} }]{
        \mathcal{G} \hsep
        \ty{\Gamma}, \ty{\Delta}, \tmty{x}{A \tens B}
      }$}
  \end{prooftree*}}
\providecommand{\hcpInfParr}{%
  \begin{prooftree*}
    \AXC{$\seq[ P ]{%
        \mathcal{G} \hsep
        \Gamma , \tmty{y}{A} , \tmty{x}{B} }$}
    \SYM{(\parr)}
    \UIC{$\seq[ \cpRecv{x}{y}{P} ]{
        \mathcal{G} \hsep
        \Gamma , \tmty{x}{A \parr B} }$}
  \end{prooftree*}}
\providecommand{\hcpInfOne}{%
  \begin{prooftree*}
    \AXC{$\seq[{ P }]{\mathcal{G}}$}
    \SYM{\one}
    \UIC{$\seq[{ \piBoundSend{x}{}{P} }]{
        \mathcal{G} \hsep \tmty{x}{\one}
      }$}
  \end{prooftree*}}
\providecommand{\hcpInfBot}{%
  \begin{prooftree*}
    \AXC{$\seq[ P ]{
        \mathcal{G} \hsep \Gamma }$}
    \SYM{(\bot)}
    \UIC{$\seq[ \cpWait{x}{P} ]{
        \mathcal{G} \hsep \Gamma , \tmty{x}{\bot} }$}
  \end{prooftree*}}
\providecommand{\hcpInfPlus}[1]{%
  \ifdim#1pt=1pt
  \begin{prooftree*}
    \AXC{$\seq[ P ]{
        \mathcal{G} \hsep
        \Gamma , \tmty{x}{A} }$}
    \SYM{(\plus_1)}
    \UIC{$\seq[{ \cpInl{x}{P} }]{
        \mathcal{G} \hsep
        \Gamma , \tmty{x}{A \plus B} }$}
  \end{prooftree*}
  \else%
  \ifdim#1pt=2pt
  \begin{prooftree*}
    \AXC{$\seq[ P ]{
        \mathcal{G} \hsep
        \Gamma , \tmty{x}{B} }$}
    \SYM{(\plus_2)}
    \UIC{$\seq[ \cpInr{x}{P} ]{
        \mathcal{G} \hsep
        \Gamma , \tmty{x}{A \plus B} }$}
  \end{prooftree*}
  \else%
  \fi%
  \fi%
}
\providecommand{\hcpInfWith}{%
  \begin{prooftree*}
    \AXC{$\seq[ P ]{
        \Gamma , \tmty{x}{A} }$}
    \AXC{$\seq[ Q ]{
        \Gamma, \tmty{x}{B} }$}
    \SYM{(\with)}
    \BIC{$\seq[ \cpCase{x}{P}{Q} ]{
        \Gamma , \tmty{x}{A \with B} }$}
  \end{prooftree*}}
\providecommand{\hcpInfNil}{\cpInfNil}
\providecommand{\hcpInfTop}{%
  \begin{prooftree*}
    \AXC{}
    \SYM{(\top)}
    \UIC{$\seq[ \cpAbsurd{x} ]{ \Gamma, \tmty{x}{\top} }$}
  \end{prooftree*}}

\usepackage[sectionbib,sort&compress,numbers]{natbib}

\title{Taking Linear Logic Apart}
\author{%
  Wen Kokke
  \institute{University of Edinburgh\\ Edinburgh, Scotland}
  \email{wen.kokke@ed.ac.uk}
  \and
  Fabrizio Montesi
  \institute{University of Southern Denmark\\ Odense, Denmark}
  \email{fmontesi@imada.sdu.dk}
  \and
  Marco Peressotti
  \institute{University of Southern Denmark\\ Odense, Denmark}
  \email{peressotti@imada.sdu.dk}}
\def\titlerunning{Taking Linear Logic Apart}
\def\authorrunning{W.\ Kokke, F.\ Montesi, and M.\ Peressotti}

\begin{document}
\maketitle

\begin{abstract}
  Process calculi based on logic, such as \piDILL and \cp, provide a foundation for deadlock-free concurrent programming. However, in previous work, there is a mismatch between the rules for constructing proofs and the term constructors of the \textpi-calculus: the fundamental operator for parallel composition does not correspond to any rule of linear logic.

  Kokke \etal~\cite{kokke2019} introduced Hypersequent Classical Processes (\dhcp), which addresses this mismatch using hypersequents (collections of sequents) to register parallelism in the typing judgements.
  However, the step from \cp to \dhcp is a big one. As of yet, \dhcp does not have reduction semantics, and the addition of delayed actions means that \cp processes interpreted as \dhcp processes do not behave as they would in \cp.

  We introduce \hcp, a variant of \dhcp with reduction semantics and without delayed actions. We prove progress, preservation, and termination, and show that \hcp supports the same communication protocols as \cp.
\end{abstract}

\section{Introduction}
\label{sec:introduction}

Classical Processes (\cp) \cite{wadler2012} is a process calculus inspired by the correspondence between the session-typed \textpi-calculus and linear logic \cite{caires2010}, where processes correspond to proofs, session types (communication protocols) to propositions, and communication to cut elimination. This correspondence allows for exchanging methods between the two fields. For example, the proof theory of linear logic can be used to guarantee progress for processes \cite{caires2010,wadler2012}.

The main attraction of \cp is that its semantics are \emph{prescribed} by the cut elimination procedure of Classical Linear Logic (CLL). This permits us to reuse the metatheory of linear logic ``as is'' to reason about the behaviour of processes. However, there is a mismatch between the structure of the proof terms of CLL and the term constructs of the standard \textpi-calculus~\cite{milner1992a,milner1992b}. For instance, the term for output of a linear name is $\tm{\cpSend{x}{y}{P}{Q}}$, which is read ``send $y$ over $x$ and proceed as $P$ in parallel to $Q$''. Note that this is a single term constructor, which takes all four arguments at the same time. This is caused by directly adopting the $(\tens)$-rule from CLL as the process calculus construct for sending: the $(\tens)$-rule has two premises (corresponding to $\tm P$ and $\tm Q$ in the output term), and checks that they share no resources (in the output term, $\tm y$ can be used only by $\tm P$, and $\tm x$ can be used only by $\tm Q$).

There is no independent parallel term $\tm{(\piPar{P}{Q})}$ in the grammar of \cp terms. Instead, parallel composition shows up in any term which corresponds to a typing rule which splits the context. Even if we were to add an independent parallel composition via the \textsc{Mix}-rule, as suggested in the original presentation of \cp~\cite{wadler2012}, there would be no way to allow the composed process $\tm{P}$ and $\tm{Q}$ to communicate as in the standard \textpi-calculus, as there is no independent name restriction either! Instead, synchronisation is governed by the ``cut'' operator $\tm{\piNew{x}{(\piPar{P}{Q})}}$, which composes $\tm{P}$ and $\tm{Q}$, enabling them to communicate along $\tm{x}$. Worse, if we naively add an independent parallel composition as well as a name restriction, using the rules shown below, we lose cut elimination, and therefore deadlock-freedom!
\[
  \begin{prooftree*}
    \AXC{$\seq[ P ]{ \Gamma }$}
    \AXC{$\seq[ Q ]{ \Delta }$}
    \NOM{Mix}
    \BIC{$\seq[ \piPar{P}{Q} ]{ \Gamma , \Delta }$}
  \end{prooftree*}
  \qquad
  \begin{prooftree*}
    \AXC{$\seq[ P ]{ \Gamma, \tmty{x}{A}, \tmty{y}{A^\bot} }$}
    \NOM{``Cut''}
    \UIC{$\seq[ \piNew{xy}{P} ]{ \Gamma }$}
  \end{prooftree*}
\]

This syntactic mismatch has an effect on the semantics as well. For instance, the $\beta$-reduction for output and input in \cp is \(\reducesto {\cpCut{x}{\cpSend{x}{y}{P}{Q}}{\cpRecv{x}{y}{R}}} {\cpCut{y}{P}{\cpCut{x}{Q}{R}}}\). Here, the parallel composition $\tm{(\piPar{P}{Q})}$ is of no relevance to this communication, yet the rule needs to inspect it to be able to nest the name restrictions appropriately in the resulting term.

Kokke \etal~\cite{kokke2019} introduced Hypersequent Classical Processes (\dhcp), which addresses this mismatch. The key insight is to register parallelism in the typing judgements using hypersequents~\cite{avron1991}, a technique from logic which generalises judgements from one sequent to many. This allows us to take apart the term constructs used in Classical Processes (\cp) to more closely match those of the standard \textpi-calculus. \dhcp has labelled transition semantics with delayed actions~\cite{merro2004} and a full abstraction result: bisimilarity, denotational equivalence, and barbed congruence coincide.

However, the step from \cp to \dhcp is a big one. As of yet, \dhcp does not have reduction semantics, and the addition of delayed actions means that \cp processes interpreted as \dhcp processes do not behave as they would in \cp.

In this paper, we address these issues by introducing \hcp, a variant of \dhcp with reduction semantics and without delayed actions. We proceed as follows. We start by introducing \cp (\cref{sec:cp}). Then, we introduce our variant of \hcp and prove it enjoys subject reduction and progress (\cref{sec:hcp}). We prove that every \cp process is an \hcp process, relate processes in \hcp back to \cp, and prove that \hcp supports the same communication protocols as \cp (\cref{sec:cp2hcp}). Finally, we discuss related work (\cref{sec:related-work}).

\section{Classical Processes}
\label{sec:cp}

In this section, we introduce \cp. In order to keep the discussion of \hcp in \cref{sec:hcp} simple, we restrict ourselves to the multiplicative-additive subset of \cp. We foresee no problems in extending the proofs in \cref{sec:hcp} to cover the remaining features of \cp (polymophism and the exponentials).

\subsection{Terms}
The term language of \cp is a variant of the \textpi-calculus. The variables $\tm{x}$, $\tm{y}$, and $\tm{z}$ range over channel names. The construct $\tm{\cpLink{x}{y}}$ links two channels~\cite{sangiorgi1996,boreale1998}, forwarding messages received on $\tm{x}$ to $\tm{y}$ and vice versa. The construct $\tm{\cpCut{x}{P}{Q}}$ creates a new channel $\tm{x}$, and composes two processes, $\tm{P}$ and $\tm{Q}$, which communicate on $\tm{x}$, in parallel. Therefore, in $\tm{\cpCut{x}{P}{Q}}$ the name $\tm{x}$ is bound in both $\tm{P}$ and $\tm{Q}$. In $\tm{\cpRecv{x}{y}{P}}$ and $\tm{\cpSend{x}{y}{P}{Q}}$, round brackets denote input, square brackets denote output. \cp uses bound output~\cite{sangiorgi1996}, meaning that both input and output bind a new name. In $\tm{\cpRecv{x}{y}{P}}$ the new name $\tm{y}$ is bound in $\tm{P}$. In $\tm{\cpSend{x}{y}{P}{Q}}$, the new name $\tm{y}$ is only bound in $\tm{P}$, while $\tm{x}$ is only bound in $\tm{Q}$.
\begin{definition}[Terms]\label{def:cp-terms}
  Process terms are given by the following grammar:
  \[
    \begin{array}{rllrll}
      \tm{P}, \tm{Q}, \tm{R}
           \Coloneqq & \tm{\cpLink{x}{y}}       &\text{link}
      &  \mid& \tm{\cpCut{x}{P}{Q}}     &\text{parallel composition, ``cut''}
      \\ \mid& \tm{\cpSend{x}{y}{P}{Q}} &\text{output}
      &  \mid& \tm{\cpRecv{x}{y}{P}}    &\text{input}
      \\ \mid& \tm{\cpHalt{x}}          &\text{halt}
      &  \mid& \tm{\cpWait{x}{P}}       &\text{wait}
      \\ \mid& \tm{\cpInl{x}{P}}        &\text{select left choice}
      &  \mid& \tm{\cpInr{x}{P}}        &\text{select right choice}
      \\ \mid& \tm{\cpCase{x}{P}{Q}}    &\text{offer binary choice}
      &  \mid& \tm{\cpAbsurd{x}}        &\text{offer nullary choice}
    \end{array}
  \]
\end{definition}\noindent
Terms in \cp are identified up to structural congruence, which states that links are symmetric, and parallel compositions $\tm{\cpCut{x}{P}{Q}}$ are associative and commutative.
\begin{definition}[Structural congruence]\label{def:cp-equiv}
  The structural congruence $\equiv$ is the congruence closure over terms which
  satisfies the following additional axioms:
  \[
    \setlength{\arraycolsep}{3pt}
    \begin{array}{llcll}
      \cpEquivLinkComm
      & \tm{\cpLink{x}{y}}
      & \equiv
      & \tm{\cpLink{y}{x}}
      \\
      \cpEquivCutComm
      & \tm{\cpCut{x}{P}{Q}}
      & \equiv
      & \tm{\cpCut{x}{Q}{P}}
      \\
      \cpEquivCutAss1
      & \tm{\cpCut{x}{P}{\cpCut{y}{Q}{R}}}
      & \equiv
      & \tm{\cpCut{y}{\cpCut{x}{P}{Q}}{R}}
      & \text{if }\notFreeIn{x}{R}\text{ and }\notFreeIn{y}{P}
    \end{array}
  \]
\end{definition}\noindent
The reduction semantics presented here are a variant of those presented by
Lindley and Morris~\cite{lindley2015}, who showed that reduction in \cp can be decomposed in two phases: one in which all \textbeta-reduction happens, and one in which an arbitrary action, blocked on an external communication, is moved to the top of the term using commuting conversions. We choose to stop after the first phase, and do away with the commuting conversions.

Reductions relate processes with their reduced forms \eg a reduction $\reducesto{P}{Q}$ denotes that the process $\tm{P}$ can reduce to the process $\tm{Q}$ in a single step.
\begin{definition}[Reduction]\label{def:cp-reduction}
  Reductions are described by the smallest relation $\Longrightarrow$ on process terms closed under rules below.
  \begin{gather*}
      \begin{array}{llcll}
        \cpRedAxCut1
        & \tm{\cpCut{x}{\cpLink{w}{x}}{P}}
        & \Longrightarrow
        & \tm{\cpSub{w}{x}{P}}
        \\
        \cpRedBetaTensParr
        & \tm{\cpCut{x}{\cpSend{x}{y}{P}{Q}}{\cpRecv{x}{y}{R}}}
        & \Longrightarrow
        & \tm{\cpCut{y}{P}{\cpCut{x}{Q}{R}}}
        \\
        \cpRedBetaOneBot
        & \tm{\cpCut{x}{\cpHalt{x}}{\cpWait{x}{P}}}
        & \Longrightarrow
        & \tm{P}
        \\
        \cpRedBetaPlusWith1
        & \tm{\cpCut{x}{\cpInl{x}{P}}{\cpCase{x}{Q}{R}}}
        & \Longrightarrow
        & \tm{\cpCut{x}{P}{Q}}
        \\
        \cpRedBetaPlusWith2
        & \tm{\cpCut{x}{\cpInr{x}{P}}{\cpCase{x}{Q}{R}}}
        & \Longrightarrow
        & \tm{\cpCut{x}{P}{R}}
      \end{array}
      \\[1\baselineskip]
      \begin{prooftree*}
        \AXC{$\reducesto{P}{P^\prime}$}
        \SYM{\cpRedGammaCut}
        \UIC{$\reducesto{\cpCut{x}{P}{Q}}{\cpCut{x}{P^\prime}{Q}}$}
      \end{prooftree*}
      \qquad
      \begin{prooftree*}
        \AXC{$\tm{P}\equiv\tm{Q}$}
        \AXC{$\reducesto{Q}{Q^\prime}$}
        \AXC{$\tm{Q^\prime}\equiv\tm{P^\prime}$}
        \SYM{\cpRedGammaEquiv}
        \TIC{$\reducesto{P}{P^\prime}$}
      \end{prooftree*}
    \end{gather*}
  Relations $\Longrightarrow^{+}$ and $\Longrightarrow^\star$ are the transitive, and the reflexive, transitive closures of $\Longrightarrow$, respectively.
\end{definition}\noindent

Note that we do not need to add a side condition to $\cpRedBetaTensParr$ to restrict its usage to the case where $\tm{y}$ is bound in $\tm{P}$ and $\tm{x}$ is bound in $\tm{Q}$, as this is required by the definition of the send construct $\tm{\cpSend{x}{y}{P}{Q}}$.

\subsection{Types}
Channels in \cp are typed using a session type system which corresponds to classical linear logic.
\begin{definition}[Types]\label{def:cp-types}
  \[
    \begin{array}{rclrcl}
      \ty{A}, \ty{B}, \ty{C}
           \Coloneqq & \ty{A \tens B} &\text{pair of independent processes}
      &  \mid& \ty{\one}      &\text{unit for} \; {\tens}
      \\ \mid& \ty{A \parr B} &\text{pair of interdependent processes}
      &  \mid& \ty{\bot}      &\text{unit for} \; {\parr}
      \\ \mid& \ty{A \plus B} &\text{internal choice}
      &  \mid& \ty{\nil}      &\text{unit for} \; {\plus}
      \\ \mid& \ty{A \with B} &\text{external choice}
      &  \mid& \ty{\top}      &\text{unit for} \; {\with}
    \end{array}
  \]
\end{definition}\noindent

A channel of type \ty{A \tens B} represents a pair of channels, which communicate with two independent processes---that is to say, two processes who share no channels. A process acting on a channel of type \ty{A \tens B} will send one endpoint of a fresh channel, and then split into a pair of independent processes. One of these processes will be responsible for an interaction of type \ty{A} over the fresh channel, while the other process continues to interact as \ty{B}.

A channel of type \ty{A \parr B} represents a pair of interdependent channels, which are used within a single process. A process acting on a channel of type \ty{A \parr B} will receive a channel to act on, and communicate on its channels in whatever order it pleases. This means that the usage of one channel can depend on that of another---\eg the interaction of type \ty{B} could depend on the result of the interaction of type \ty{A}, or vise versa, and if \ty{A} and \ty{B} are complex types, their interactions could likewise interweave in complex ways.

A process acting on a channel of type \ty{A \plus B} either sends the value \tm{\text{inl}} to select an interaction of type \ty{A} or the value \tm{\text{inr}} to select one of type \ty{B}. 
A process acting on a channel of type \ty{A \with B} receives such a value, and then offers an interaction of either type \ty{A} or \ty{B}, correspondingly.

Duality plays a crucial role in both linear logic and session types. In \cp, the two endpoints of a channel are assigned dual types. This ensures that, for instance, whenever a process \emph{sends} across a channel, the process on the other end of that channel is waiting to \emph{receive}. Each type $\ty{A}$ has a dual, written $\ty{A^\bot}$. Duality is an involutive function \ie $\ty{(A^\bot)^\bot} = \ty{A}$.
\begin{definition}[Duality]\label{def:cp-negation}
  \[
    \setlength{\arraycolsep}{3pt}
    \begin{array}{lclclclclclclcl}
               \ty{(A \tens B)^\bot} & = & \ty{A^\bot \parr B^\bot}
      &\qquad& \ty{\one^\bot}        & = & \ty{\bot}
      &\qquad& \ty{(A \parr B)^\bot} & = & \ty{A^\bot \tens B^\bot}
      &\qquad& \ty{\bot^\bot}        & = & \ty{\one}
      \\       \ty{(A \plus B)^\bot} & = & \ty{A^\bot \with B^\bot}
      &\qquad& \ty{\nil^\bot}        & = & \ty{\top}
      &\qquad& \ty{(A \with B)^\bot} & = & \ty{A^\bot \plus B^\bot}
      &\qquad& \ty{\top^\bot}        & = & \ty{\nil}
    \end{array}
  \]
\end{definition}\noindent
An environment associates channels with types. Names in environments must be unique, and two environments $\ty{\Gamma}$ and $\ty{\Delta}$ can only be combined as $\ty{\Gamma, \Delta}$ if $\ty{\fv{\Gamma}} \cap \ty{\fv{\Delta}} = \varnothing$.
\begin{definition}[Environments]\label{def:cp-environments}
  \(
    \ty{\Gamma}, \ty{\Delta}, \ty{\Theta} \Coloneqq \ty{\emptycontext} \mid \ty{\Gamma}, \tmty{x}{A}
  \)
\end{definition}\noindent
A typing judgement associates a process with collections of typed channels.
\begin{definition}[Typing judgements]\label{def:cp}
  A typing judgement $\seq[{ P }]{\tmty{x_1}{A_1}, \dots, \tmty{x_n}{A_n}}$
  denotes that the process \tm{P} communicates along channels $\tm{x_1}$, \dots,
  $\tm{x_n}$ following protocols $\ty{A_1}$, \dots, $\ty{A_n}$. Typing
  judgements are derived using rules below.
  \\[1\baselineskip]
  {Structural rules}
  \begin{center}
    \cpInfAx \cpInfCut
  \end{center}
	{Logical rules}
  \begin{center}
    \cpInfTens \cpInfParr
  \end{center}
  \begin{center}
    \cpInfBot
    \cpInfOne
  \end{center}
  \begin{center}
    \cpInfPlus1
    \cpInfPlus2
  \end{center}
  \begin{center}
    \cpInfWith
  \end{center}
  \begin{center}
    \cpInfNil
    \cpInfTop   
  \end{center}
\end{definition}

\subsection{Metatheory}
\cp enjoys subject reduction, termination, and progress~\cite{wadler2012,lindley2015}.
\begin{lemma}[Preservation for $\equiv$]\label{lem:cp-preservation-equiv}
  If $\tm{P}\equiv\tm{Q}$, then $\seq[P]{\Gamma}$ iff $\seq[Q]{\Gamma}$.
\end{lemma} 
\begin{proof}
  By induction on the derivation of $\tm{P}\equiv\tm{Q}$.
\end{proof}
\begin{theorem}[Preservation]\label{thm:cp-preservation}
  If $\seq[P]{\Gamma}$ and $\reducesto{P}{Q}$, then $\seq[Q]{\Gamma}$.
\end{theorem} 
\begin{proof}
  By induction on the derivation of $\reducesto{P}{Q}$.
\end{proof}
\begin{definition}[Actions]
  A process $\tm{P}$ acts on $\tm{x}$ whenever $\tm{x}$ is free in the outermost
  term constructor of $\tm{P}$, \eg $\tm{\cpSend{x}{y}{P}{Q}}$ acts on $\tm{x}$
  but not on $\tm{y}$, and $\tm{\cpLink{x}{y}}$ acts on both $\tm{x}$ and $\tm{y}$.
  A process $\tm{P}$ is an action if it acts on some channel $\tm{x}$.
\end{definition}
\begin{definition}[Canonical forms]
  A process $\tm{P}$ is in canonical form if
  \[
    \tm{P} \equiv \tm{\cpCut{x_1}{P_1}{\dots\cpCut{x_n}{P_n}{P_{n+1}}\dots}},
  \]
  such that: no process $\tm{P_i}$ is a cut; no process $\tm{P_i}$ is a link acting on a bound channel $\tm{x_i}$; and no two processes $\tm{P_i}$ and $\tm{P_j}$ are acting on the same bound channel $\tm{x_i}$.
\end{definition}
\begin{corollary}
  If a process $\tm{P}$ is in canonical form, then it is blocked on an external communication.
\end{corollary}
\begin{proof}
  We have
  \[
  \tm{P} \equiv \tm{\cpCut{x_1}{P_1}{\dots\cpCut{x_n}{P_n}{P_{n+1}}\dots}}
  \]
  such that no $\tm{P_i}$ is a cut or a link, and no two processes $\tm{P_i}$ and $\tm{P_j}$ are acting on the same bound channel. The prefix of cuts introduces $n$ channels, and $n+1$ processes. Therefore, at least \emph{one} of the processes $\tm{P_i}$ must be acting on a free channel, i.e., blocked on an external communication.
\end{proof}
\begin{theorem}[Progress]\label{thm:cp-progress}
  If $\seq[P]{\Gamma}$, then either $\tm{P}$ is in canonical form, or there exists a process $\tm{Q}$ such that $\tm{P}\Longrightarrow\tm{Q}$.
\end{theorem} 
\begin{proof}
  We consider the maximum prefix of cuts of $\tm{P}$ such that $\tm{P} \equiv \tm{\cpCut{x_1}{P_1}{\dots\cpCut{x_n}{P_n}{P_{n+1}}\dots}}$ and no $\tm{P_i}$ is a cut. If any process $\tm{P_i}$ is a link, we reduce by $(\cpLink{}{})$. If any two processes $\tm{P_i}$ and $\tm{P_j}$ are acting on the same channel $\tm{x_i}$, we rewrite by $\equiv$ and reduce by the appropriate $\beta$-rule. Otherwise, $\tm{P}$ is in canonical form.
\end{proof}
\begin{theorem}[Termination]\label{thm:cp-termination}
  If $\seq[P]{\Gamma}$, then there are no infinite $\Longrightarrow$-reduction sequences.
\end{theorem} 
\begin{proof}
  Every reduction reduces a single cut to zero, one or two cuts. However, each of these cuts is smaller, measured in the size of the cut formula. Furthermore, each instance of the structural congruence preserves the size of the cut. Therefore, there cannot be an infinite $\Longrightarrow$-reduction sequence.
\end{proof}

\section{Hypersequent Classical Processes}
\label{sec:hcp}

In this section, we introduce our variant of Hypersequent Classical Processes (\hcp), itself a variant of \cp which registers parallelism in the typing judgements using hypersequents, allowing us to take apart the monolithic term constructors of \cp (\eg $\tm{\cpSend{x}{y}{P}{Q}}$) into the corresponding \textpi-calculus term constructs.

The crucial difference between \hcp as described here and \dhcp as described by Kokke \etal~\cite{kokke2019} is in the absence of delayed actions. However, removing delayed actions introduces self-locking processes, which we rule out using an extra restriction in the type system (see \cref{sec:hcp-types}). Furthermore, we follow \cp in using the same name for both endpoints of a channel, writing, \eg $\tm{\cpCut{x}{\cpHalt{x}}{\cpWait{x}{P}}}$ as opposed to $\tm{\cpCut{xy}{\cpHalt{x}}{\cpWait{y}{P}}}$.

\subsection{Terms}
The term language of \hcp is a variant of \cp where the term constructs have been taken apart into primitives which more closely resemble the \textpi-calculus primitives.
\begin{definition}[Terms]\label{def:hcp-terms}
  \[
    \begin{array}{rllrll}
      \tm{P}, \tm{Q}, \tm{R}
           \Coloneqq & \tm{\cpLink{x}{y}}         &\text{link}
      &  \mid& \tm{\piHalt}               &\text{terminated process}
      \\ \mid& \tm{\piNew{x}{P}}          &\text{name restriction, ``cut''}
      &  \mid& \tm{( \piPar{P}{Q} )}      &\text{parallel composition, ``mix''}
      \\ \mid& \tm{\piBoundSend{x}{y}{P}} &\text{output}
      &  \mid& \tm{\piRecv{x}{y}{P}}      &\text{input}
      \\ \mid& \tm{\piBoundSend{x}{}{P}}  &\text{halt}
      &  \mid& \tm{\cpWait{x}{}{P}}       &\text{wait}
      \\ \mid& \tm{\cpInl{x}{P}}          &\text{select left choice}
      &  \mid& \tm{\cpInr{x}{P}}          &\text{select right choice}
      \\ \mid& \tm{\cpCase{x}{P}{Q}}      &\text{offer binary choice}
      &  \mid& \tm{\cpAbsurd{x}}          &\text{offer nullary choice}
    \end{array}
  \]
\end{definition}\noindent
A pleasant effect of our updated syntax is that it makes our structural congruence much more standard: it has associativity, commutativity, and a unit for parallel composition, commutativity of name restrictions, and scope extrusion.
\begin{definition}[Structural congruence]\label{def:hcp-equiv}
  The structural congruence $\equiv$ is the congruence closure over terms which satisfies the following additional axioms:
  \[
    \setlength{\arraycolsep}{3pt}
    \begin{array}{llcllllcll}
        \cpEquivLinkComm
      & \tm{\cpLink{x}{y}}
      & \equiv
      & \tm{\cpLink{y}{x}}
      &
      &
        \hcpEquivMixHalt1
      & \tm{\piPar{P}{\piHalt}}
      & \equiv
      & \tm{P}
      &
      \\
        \hcpEquivMixComm
      & \tm{\piPar{P}{Q}}
      & \equiv
      & \tm{\piPar{Q}{P}}
      &
      &
        \hcpEquivMixAss1
      & \tm{\piPar{P}{( \piPar{Q}{R} )}}
      & \equiv
      & \tm{\piPar{( \piPar{P}{Q} )}{R}}
      &
      \\
        \hcpEquivNewComm
      & \tm{\piNew{x}{\piNew{y}{P}}}
      & \equiv
      & \tm{\piNew{y}{\piNew{x}{P}}}
      &
      &
        \hcpEquivScopeExt1
      & \tm{\piNew{x}{( \piPar{P}{Q} )}}
      & \equiv
      & \tm{\piPar{P}{\piNew{x}{Q}}}
      & \text{if }\notFreeIn{x}{P}
    \end{array}
  \]
\end{definition}\noindent
There are two changes to the reduction system. First, since $\tm{\piBoundSend{x}{y}{P}}$ and $\tm{\piBoundSend{x}{}{P}}$ are now terms in their own right, the $\cpRedBetaTensParr$ and $\cpRedBetaOneBot$ rules are simpler. Second, since we decomposed $\tm{\cpCut{x}{P}{Q}}$ into an independent name restriction and parallel composition, the relevant $\gamma$-rule all decompose as well.
\begin{definition}[Reduction]\label{def:hcp-reduction}
  Reductions are described by the smallest relation $\Longrightarrow$ on process
  terms closed under the rules below:
  \begin{gather*}
    \begin{array}{llcll}
      \hcpRedAxCut1
      & \tm{\cpCut{x}{\cpLink{w}{x}}{P}}
      & \Longrightarrow
      & \tm{\cpSub{w}{x}{P}}
      \\
      \hcpRedBetaTensParr
      & \tm{\cpCut{x}{\piBoundSend{x}{y}{P}}{\piRecv{x}{y}{R}}}
      & \Longrightarrow
      & \tm{\piNew{x}{\piNew{y}{(\piPar{P}{R})}}}
      \\
      \hcpRedBetaOneBot
      & \tm{\cpCut{x}{\piBoundSend{x}{}{P}}{\piRecv{x}{}{Q}}}
      & \Longrightarrow
      & \tm{\piPar{P}{Q}}
      \\
      \hcpRedBetaPlusWith1
      & \tm{\cpCut{x}{\cpInl{x}{P}}{\cpCase{x}{Q}{R}}}
      & \Longrightarrow
      & \tm{\cpCut{x}{P}{Q}}
      \\
      \hcpRedBetaPlusWith2
      & \tm{\cpCut{x}{\cpInr{x}{P}}{\cpCase{x}{Q}{R}}}
      & \Longrightarrow
      & \tm{\cpCut{x}{P}{R}}
    \end{array}
    \\[1\baselineskip]
    \begin{prooftree*}
      \AXC{$\reducesto{P}{P^\prime}$}
      \SYM{\hcpRedGammaNew}
      \UIC{$\reducesto{\piNew{x}{P}}{\piNew{x}{P^\prime}}$}
    \end{prooftree*}
    \qquad
    \begin{prooftree*}
      \AXC{$\reducesto{P}{P^\prime}$}
      \SYM{\hcpRedGammaMix}
      \UIC{$\reducesto{\piPar{P}{Q}}{\piPar{P^\prime}{Q}}$}
    \end{prooftree*}
    \qquad
    \begin{prooftree*}
      \AXC{$\tm{P}\equiv\tm{Q}$}
      \AXC{$\reducesto{Q}{Q^\prime}$}
      \AXC{$\tm{Q^\prime}\equiv\tm{P^\prime}$}
      \SYM{\hcpRedGammaEquiv}
      \TIC{$\reducesto{P}{P^\prime}$}
    \end{prooftree*}
  \end{gather*}  
  Relations $\Longrightarrow^{+}$ and $\Longrightarrow^\star$ are the transitive, and the reflexive, transitive closures of $\Longrightarrow$, respectively.
\end{definition}

\subsection{Types}\label{sec:hcp-types}
We use the same definitions for types and environments for \hcp as we used for \cp. However, we introduce a new layer on top of sequents: hypersequents. As \cp is a one-sided logic, and it uses the left-hand side of the turnstile to write the process, the traditional hypersequent notation can look confusing: ``$\seq[P]{\Gamma_1}\hsep\dots\hsep\seq{\Gamma_n}$'' seems to claim that $\tm{P}$ acts according to protocol $\ty{\Gamma_1}$. What are all the other $\Gamma$s doing there? Are they typing empty processes? Therefore, we opt to leave out the repeated turnstile, and instead work with the notion of ``hyper-environments''. However, we will still refer to our system as a hypersequent system. A hyper-environment is either empty, or consist of a series of typing environments, separated by vertical bars. A hyper-environment $\ty{\Gamma_1 \hsep \dots \hsep \Gamma_n}$ types a series of $n$ entangled, but independent processes.
\begin{definition}[Hyper-environments]\label{def:hcp-hyper-environment}
  $\ty{\mathcal{G}}, \ty{\mathcal{H}} \Coloneqq \ty{\emptyhypercontext} \, \mid \, \ty{\mathcal{G} \hsep \Gamma}$
\end{definition}\noindent
A hyper-environment is a multiset of environments. While names within environments must be unique, names may be shared between multiple environments in a hyper-environment. We write $\ty{\mathcal{G} \hsep \mathcal{H}}$ to combine two hyper-environments.

Typing judgements in \hcp associate processes with hyper-environments. \textsc{H-Mix} composes two processes in parallel, but remembers that they are independent in the sequent. \textsc{H-Cut} and $(\tens)$ take as their premise a process which consists of at least two independent processes, and connects them, eliminating the vertical bar. Each logical rule has the side condition that $\tm{x} \not\in \ty{\mathcal{G}}$, which can be read as ``you cannot act on one end-point of $x$ if you are also holding its other end-point''. This prevents self-locking processes, \eg $\tm{\piRecv{x}{}{\piBoundSend{x}{}{\piHalt}}}$.
\begin{definition}[Typing judgements]\label{def:hcp}
  A typing judgement $\seq[P]{\Gamma_1 \hsep \dots \hsep \Gamma_n}$ denotes that the process $\tm{P}$ consists of $n$ independent, but potentially entangled processes, each of which communicates according to its own protocol $\Gamma_i$. 
  Typing judgements can be constructed using the inference rules below. 
  \\[1\baselineskip]
  {Structural rules}
  \begin{center}
    \hcpInfAx
    \hcpInfCut
  \end{center}
  \begin{center}
    \hcpInfMix
    \hcpInfHalt
  \end{center}
  {Logical rules}
  \begin{center}
    \hcpInfBoundTens
    \hcpInfParr
  \end{center}
  \begin{center}
    \hcpInfOne
    \hcpInfBot
  \end{center}
  \begin{center}
    \hcpInfPlus1
    \hcpInfPlus2
  \end{center}
  \begin{center}
    \hcpInfWith
  \end{center}
  \begin{center}
    \hcpInfNil
    \hcpInfTop
  \end{center}
  Furthermore, each logical rule has the side condition that $\tm{x} \not\in \ty{\mathcal{G}}$.
\end{definition}\noindent

Note that the rules $(\with)$ and $(\top)$ disallow hyperenvironments. Allowing hyperenvironments in $(\with)$ would allow us to derive processes which are stuck. For instance, the following process would be well-typed, but is stuck---even in the presence of delayed actions:
\[
\seq[{
  \tm{\piNew{x^{\ty{\one \with \one}}}{\piNew{y^{\ty{\bot \plus \bot}}}{\left(
          \begin{array}{l}
            x \triangleright
            \left\{
            \begin{array}{l}
              \texttt{inl}: (\cpInl{y}{\cpWait{y}{\cpHalt{z}}} \ppar \cpHalt{x});
              \\
              \texttt{inr}: (\cpInr{y}{\cpWait{y}{\cpHalt{z}}} \ppar \cpHalt{x})
            \end{array}
            \right\}
            \ppar
            \\
            y \triangleright
            \left\{
            \begin{array}{l}
              \texttt{inl}: (\cpInl{x}{\cpWait{x}{\cpHalt{w}}} \ppar \cpHalt{y});
              \\
              \texttt{inr}: (\cpInr{x}{\cpWait{x}{\cpHalt{w}}} \ppar \cpHalt{y})
            \end{array}
            \right\}
          \end{array}
        \right)}}}
}]{
  \tmty{z}{\one} \hsep \tmty{w}{\one}
}
\]
Allowing hypersequents in $(\top)$ would lead to problems with \cref{lem:hcp-disentangle}. Intuitively, if we allowed the derivation $\seq[\cpAbsurd{x}]{\mathcal{G} \hsep \Gamma, \tmty{x}{\top}}$, we would claim that $\tm{\cpAbsurd{x}}$ ``consistst of $n$ independent, but potentially entangled processes'', which is clearly false.

\subsection{Metatheory}
\hcp enjoys subject reduction, termination, and progress.
\begin{lemma}[Preservation for $\equiv$]\label{lem:hcp-preservation-equiv}
  If $\tm{P}\equiv\tm{Q}$, then $\seq[P]{\mathcal{G}}$ iff $\seq[Q]{\mathcal{G}}$.
\end{lemma} 
\begin{proof}
  By induction on the derivation of $\tm{P}\equiv\tm{Q}$.
\end{proof}
\begin{theorem}[Preservation]\label{thm:hcp-preservation}
  If $\seq[P]{\mathcal{G}}$ and $\reducesto{P}{Q}$, then $\seq[Q]{\mathcal{G}}$.
\end{theorem} 
\begin{proof}
  By induction on the derivation of $\reducesto{P}{Q}$.
\end{proof}
\begin{definition}[Actions]
  A process $\tm{P}$ acts on $\tm{x}$ whenever $\tm{x}$ is free in the outermost
  term constructor of $\tm{P}$, \eg $\tm{\cpSend{x}{y}{P}{Q}}$ acts on $\tm{x}$
  but not on $\tm{y}$, and $\tm{\cpLink{x}{y}}$ acts on both $\tm{x}$ and $\tm{y}$.
  A process $\tm{P}$ is an action if it acts on some channel $\tm{x}$.
\end{definition}
\begin{definition}[Canonical forms]
  A process $\tm{P}$ is in canonical form if
  \[
  \tm{P} \equiv \tm{\piNew{x_1}{\dots\piNew{x_n}{(P_1 \mid \dots \mid P_{n+m+1})}}},
  \]
  such that: no process $\tm{P_i}$ is a cut or a mix; no process $\tm{P_i}$ is a link acting on a bound channel $\tm{x_i}$; and no two processes $\tm{P_i}$ and $\tm{P_j}$ are acting on the same bound channel $\tm{x_i}$.
\end{definition}
Note that we have added the restriction ``acting on a bound channel'' to the case for links. This was not necessary for \cp, as all links in \cp act on at least one bound channel. Consequently, processes such as $\tm{\cpLink{x}{y}}$ and $\tm{(\piPar{\cpLink{x}{y}}{\cpLink{z}{w}})}$ are considered to be in canonical form. This is a generalisation of \cp, where $\tm{\cpLink{x}{y}}$ is considered to be in canonical form. If this is objectionable, the reduction system can be extended with identity expansion, expanding, \eg the process $\tm{\cpLink{x^{\ty{\bot}}}{y}}$ to $\tm{\cpWait{x}{\cpHalt{y}}}$.
\begin{corollary}
  If a process $\tm{P}$ is in canonical form, then it is blocked on an external communication.
\end{corollary}
\begin{proof}
  We have
  \[
  \tm{P} \equiv \tm{\piNew{x_1}{\dots\piNew{x_n}{(P_1 \ppar \dots \ppar P_{n+m+1})}}},
  \]
  such that no $\tm{P_i}$ is a cut or a link acting on a bound channel, and no two processes $\tm{P_i}$ and $\tm{P_j}$ are acting on the same bound channel. The prefix of cuts and mixes introduces $n$ channels. Each application of cut requires an application of mix, so the prefix introduces $n+m+1$ processes. Therefore, at least $m+1$ of the processes $\tm{P_i}$ must be acting on a free channel, i.e., blocked on an external communication.
\end{proof}
\begin{theorem}[Progress]\label{thm:hcp-progress}
  If $\seq[P]{\Gamma}$, then either $\tm{P}$ is in canonical form, or there exists a process $\tm{Q}$ such that $\tm{P}\Longrightarrow\tm{Q}$.
\end{theorem} 
\begin{proof}
  We consider the maximum prefix of cuts and mixes of $\tm{P}$ such that
  \[
  \tm{P} \equiv \tm{\piNew{x_1}{\dots\piNew{x_n}{(P_1 \ppar \dots \ppar P_{n+m+1})}}},
  \]
  and no $\tm{P_i}$ is a cut. If any process $\tm{P_i}$ is a link, we reduce by $(\cpLink{}{})$. If any two processes $\tm{P_i}$ and $\tm{P_j}$ are acting on the same channel $\tm{x_i}$, we rewrite by $\equiv$ and reduce by the appropriate $\beta$-rule. Otherwise, $\tm{P}$ is in canonical form.
\end{proof}
\begin{theorem}[Termination]\label{thm:hcp-termination}
  If $\seq[P]{\mathcal{G}}$, then there are no infinite $\Longrightarrow$-reduction sequences.
\end{theorem} 
\begin{proof}
  As \cref{thm:cp-termination}.
\end{proof}

\section{Relation between \cp and \hcp}
\label{sec:cp2hcp}

In this section, we discuss the relationship between \cp and \hcp. We prove two important theorems: every \cp process is an \hcp process; and \hcp supports the same protocols as \cp. We define a translation from terms in \cp to terms in \hcp which breaks down the term constructs in \cp into their more atomic constructs in \hcp.
\begin{definition}\label{def:cp2hcp-terms}
  \[
    \begin{array}{lcllcl}
         \tm{\mtf{\cpLink{x}{y}}}
      &  \coloneqq & \tm{\cpLink{x}{y}}
      &  \tm{\mtf{\cpCut{x}{P}{Q}}}
      &  \coloneqq & \tm{\piNew{x}{(\piPar{\mtf{P}}{\mtf{Q}})}}
      \\ \tm{\mtf{\cpSend{x}{y}{P}{Q}}}
      &  \coloneqq & \tm{\piBoundSend{x}{y}{(\piPar{\mtf{P}}{\mtf{Q}})}}
      &  \tm{\mtf{\cpRecv{x}{y}{P}}}
      &  \coloneqq & \tm{\piRecv{x}{y}{\mtf{P}}}
      \\ \tm{\mtf{\cpHalt{x}}}
      &  \coloneqq & \tm{\piBoundSend{x}{}{\piHalt}}
      &  \tm{\mtf{\cpWait{x}{P}}}
      &  \coloneqq & \tm{\piRecv{x}{}{\mtf{P}}}
      \\ \tm{\mtf{\cpInl{x}{P}}}
      &  \coloneqq & \tm{\cpInl{x}{\mtf{P}}}
      &  \tm{\mtf{\cpInr{x}{P}}}
      &  \coloneqq & \tm{\cpInr{x}{\mtf{P}}}
      \\ \tm{\mtf{\cpCase{x}{P}{Q}}}
      &  \coloneqq & \tm{\cpCase{x}{\mtf{P}}{\mtf{Q}}}
      &  \tm{\mtf{\cpAbsurd{x}}}
      &  \coloneqq & \tm{\cpAbsurd{x}}
    \end{array}
  \]
\end{definition}\noindent
We use this relation in the first proof, and its analogue for derivations in the second.

\subsection{Every \cp process is an \hcp process}
First, we prove that each \cp process can be translated by this trivial translation to an \hcp process, and that this translation respects structural congruence and reduction. Reductions from \cp can be trivially translated to reductions in \hcp. 
\begin{theorem}\label{thm:cp2hcp-typing}
  If $\seq[P]{\Gamma}$ in \cp, then $\seq[\mtf{P}]{\Gamma}$ in \hcp.
\end{theorem} 
\begin{proof}
  By induction on the derivation of $\seq[P]{\Gamma}$. We show the interesting cases:
  \begin{itemize}
  \item
    Case \textsc{Cut}. We rewrite as follows:
    \begin{gather*}
      \begin{array}{lcl}
        \AXC{$\seq[P]{\Gamma, \tmty{x}{A}}$}
        \AXC{$\seq[Q]{\Delta, \tmty{x}{A^\bot}}$}
        \NOM{Cut}
        \BIC{$\seq[\cpCut{x}{P}{Q}]{\Gamma, \Delta}$}
        \DisplayProof
        & \Rightarrow
        & \AXC{$\seq[\mtf{P}]{\Gamma, \tmty{x}{A}}$}
          \AXC{$\seq[\mtf{Q}]{\Delta, \tmty{x}{A^\bot}}$}
          \NOM{H-Mix}
          \BIC{$\seq[\piPar{\mtf{P}}{\mtf{Q}}]{
          \Gamma, \tmty{x}{A} \hsep \Delta, \tmty{x}{A^\bot}}$}
          \NOM{H-Cut}
          \UIC{$\seq[\piNew{x}{(\piPar{\mtf{P}}{\mtf{Q}})}]{
          \Gamma, \Delta}$}
          \DisplayProof
      \end{array}
    \end{gather*}
  \item
    Case $(\tens)$. We rewrite as follows:
    \begin{gather*}
      \begin{array}{lcl}
        \AXC{$\seq[P]{\Gamma, \tmty{y}{A}}$}
        \AXC{$\seq[Q]{\Delta, \tmty{x}{B}}$}
        \SYM{\tens}
        \BIC{$\seq[\cpSend{x}{y}{P}{Q}]{\Gamma, \Delta, \tmty{x}{A \tens B}}$}
        \DisplayProof
        & \Rightarrow
        & \AXC{$\seq[\mtf{P}]{\Gamma, \tmty{y}{A}}$}
          \AXC{$\seq[\mtf{Q}]{\Delta, \tmty{x}{B}}$}
          \NOM{H-Mix}
          \BIC{$\seq[\piPar{\mtf{P}}{\mtf{Q}}]{
          \Gamma, \tmty{y}{A} \hsep \Delta, \tmty{x}{B}}$}
          \SYM{\tens}
          \UIC{$\seq[\piBoundSend{x}{y}{(\piPar{\mtf{P}}{\mtf{Q}})}]{
          \Gamma, \Delta, \tmty{x}{A \tens B}}$}
          \DisplayProof
      \end{array}
    \end{gather*}
  \item
    Case $(\one)$. We rewrite as follows:
    \begin{gather*}
      \begin{array}{lcl}
        \AXC{}\SYM{\one}
        \UIC{$\seq[\cpHalt{x}]{\tmty{x}{\one}}$}
        \DisplayProof
        & \Rightarrow
        & \AXC{}
          \NOM{H-Mix$_0$}
          \UIC{$\seq[\piHalt]{\emptyhypercontext}$}
          \SYM{\one}
          \UIC{$\seq[\piBoundSend{x}{}{\piHalt}]{\tmty{x}{\one}}$}
          \DisplayProof
      \end{array}
    \end{gather*}
  \end{itemize}
  \vspace*{-\baselineskip}
\end{proof}
\begin{theorem}\label{thm:cp2hcp-equiv}
  If $\tm{P}\equiv\tm{Q}$ in \cp, then $\tm{\mtf{P}}\equiv\tm{\mtf{Q}}$ in \hcp.
\end{theorem} 
\begin{proof}
  By induction on the derivation of $\tm{P}\equiv\tm{Q}$.
\end{proof}
\begin{theorem}\label{thm:cp2hcp-reduction}
  If $\reducesto{P}{Q}$ in \cp, then $\reducesto{\mtf{P}}{\mtf{Q}}$ in \hcp.
\end{theorem} 
\begin{proof}
  By induction on the the derivation of $\reducesto{P}{Q}$.
\end{proof}
\begin{theorem}\label{hcp2cp-reduction}
  If $\reducesto{\mtf{P}}{R}$ in \hcp, then there is a $\tm{Q}$ such that $\reducesto{P}{Q}$ in \cp and $\tm{R}\equiv\tm{\mtf{Q}}$ in \hcp.

\end{theorem} 
\begin{proof}
  By induction on the derivation of $\reducesto{\mtf{P}}{R}$.
\end{proof}\noindent

\subsection{\hcp supports the same communication protocols as \cp}
In this section, we prove that \hcp supports the same communication protocols as \cp. This is the same as saying that it inhabits the same session types, or that the associated logical systems derive the same theorems. We show this by proving that we can internalise the hyper-environments as formulas in the logic. This is a standard method for proving the soundness of a hypersequent calculus.

We start off by defining a relation on derivations of \hcp, which we call ``disentanglement''. This relation allows us to move applications of $\textsc{H-Mix}$ downwards in the proof tree. We can use this relation to rewrite any derivation to a form in which all mixes are either attached to their respective cuts or tensors, or at the top-level.
\begin{definition}\label{def:hcp-disentangle}
  Disentanglement is described by the smallest relation $\moveMixDown$ on processes closed under the rules in \cref{fig:hcp-disentangle}, plus the structural congruence $\equiv$. The relation $\moveMixDown^\star$ is the reflexive, transitive closure of $\moveMixDown$.
\end{definition}\noindent
\begin{sidewaysfigure} 
  \[
  \begin{array}{lcr}
    \AXC{$\seq[P]{\mathcal{G} \hsep \Gamma, \tmty{x}{A} \hsep \Delta, \tmty{x}{A^\bot}}$}
    \AXC{$\seq[Q]{\mathcal{H}}$}
    \NOM{H-Mix}
    \BIC{$\seq[(P \ppar Q)]{\mathcal{G} \hsep \mathcal{H} \hsep \Gamma, \tmty{x}{A} \hsep \Delta, \tmty{x}{A^\bot}}$}
    \NOM{H-Cut}
    \UIC{$\seq[\piNew{x}{(P \ppar Q)}]{\mathcal{G} \hsep \mathcal{H} \hsep \Gamma, \Delta}$}
    \DisplayProof
    & \moveMixDown
    & \AXC{$\seq[P]{\mathcal{G} \hsep \Gamma, \tmty{x}{A} \hsep \Delta, \tmty{x}{A^\bot}}$}
      \NOM{H-Cut}
      \UIC{$\seq[\piNew{x}{P}]{\mathcal{G} \hsep \Gamma, \Delta}$}
      \AXC{$\seq[Q]{\mathcal{H}}$}
      \NOM{H-Mix}
      \BIC{$\seq[(\piNew{x}{P} \ppar Q)]{\mathcal{G} \hsep \mathcal{H} \hsep \Gamma, \Delta}$}
      \DisplayProof
    \\\\
    \AXC{$\seq[P]{\mathcal{G} \hsep \Gamma, \tmty{y}{A} \hsep \Delta, \tmty{x}{B}}$}
    \AXC{$\seq[Q]{\mathcal{H}}$}
    \NOM{H-Mix}
    \BIC{$\seq[(P \ppar Q)]{\mathcal{G} \hsep \mathcal{H} \hsep \Gamma, \tmty{y}{A} \hsep \Delta, \tmty{x}{B}}$}
    \SYM{\tens}
    \UIC{$\seq[\piBoundSend{x}{y}{(P \ppar Q)}]{\mathcal{G} \hsep \mathcal{H} \hsep \Gamma, \Delta, \tmty{x}{A \tens B}}$}
    \DisplayProof
    & \moveMixDown
    & \AXC{$\seq[P]{\mathcal{G} \hsep \Gamma, \tmty{y}{A} \hsep \Delta, \tmty{x}{B}}$}
      \SYM{\tens}
      \UIC{$\seq[\piBoundSend{x}{y}{P}]{\mathcal{G} \hsep \Gamma, \Delta, \tmty{x}{A \tens B}}$}
      \AXC{$\seq[Q]{\mathcal{H}}$}
      \NOM{H-Mix}
      \BIC{$\seq[(\piBoundSend{x}{y}{P} \ppar Q)]{\mathcal{G} \hsep \mathcal{H} \hsep \Gamma, \Delta, \tmty{x}{A \tens B}}$}
      \DisplayProof
    \\\\
    \AXC{$\seq[P]{\mathcal{G} \hsep \Gamma, \tmty{y}{A}, \tmty{x}{B}}$}
    \AXC{$\seq[Q]{\mathcal{H}}$}
    \NOM{H-Mix}
    \BIC{$\seq[(P \ppar Q)]{\mathcal{G} \hsep \mathcal{H} \hsep \Gamma, \tmty{y}{A}, \tmty{x}{B}}$}
    \SYM{\parr}
    \UIC{$\seq[(\piRecv{x}{y}{(P \ppar Q)})]{\mathcal{G} \hsep \mathcal{H} \hsep \Gamma, \tmty{x}{A \parr B}}$}
    \DisplayProof
    & \moveMixDown
    & \AXC{$\seq[P]{\mathcal{G} \hsep \Gamma, \tmty{y}{A}, \tmty{x}{B}}$}
      \SYM{\parr}
      \UIC{$\seq[\piRecv{x}{y}{P}]{\mathcal{G} \hsep \mathcal{H} \hsep \Gamma, \tmty{x}{A \parr B}}$}
      \AXC{$\seq[Q]{\mathcal{H}}$}
      \NOM{H-Mix}
      \BIC{$\seq[(\piRecv{x}{y}{P} \ppar Q)]{\mathcal{G} \hsep \mathcal{H} \hsep \Gamma, \tmty{x}{A \parr B}}$}
      \DisplayProof
    \\\\
    \AXC{$\seq[P]{\mathcal{G}}$}
    \SYM{\one}
    \UIC{$\seq[\piBoundSend{x}{}{P}]{\mathcal{G} \hsep \tmty{x}{\one}}$}
    \DisplayProof
    & \moveMixDown
    & \AXC{}
      \NOM{H-Mix$_0$}
      \UIC{$\seq[\piHalt]{\emptyhypercontext}$}
      \SYM{\one}
      \UIC{$\seq[\piBoundSend{x}{}{0}]{\tmty{x}{\one}}$}
      \AXC{$\seq[P]{\mathcal{G}}$}
      \NOM{H-Mix}
      \BIC{$\seq[(\piBoundSend{x}{}{\piHalt} \ppar P)]{\mathcal{G} \hsep \tmty{x}{\one}}$}
      \DisplayProof
    \\\\
    \AXC{$\seq[P]{\mathcal{G} \hsep \Gamma}$}
    \AXC{$\seq[Q]{\mathcal{H}}$}
    \NOM{H-Mix}
    \BIC{$\seq[(P \ppar Q)]{\mathcal{G} \hsep \mathcal{H} \hsep \Gamma}$}
    \SYM{\bot}
    \UIC{$\seq[\piRecv{x}{}{(P \ppar Q)}]{\mathcal{G} \hsep \mathcal{H} \hsep \Gamma, \tmty{x}{\bot}}$}
    \DisplayProof
    & \moveMixDown
    & \AXC{$\seq[P]{\mathcal{G} \hsep \Gamma}$}
      \SYM{\bot}
      \UIC{$\seq[\piRecv{x}{}{P}]{\mathcal{G} \hsep \Gamma, \tmty{x}{\bot}}$}
      \AXC{$\seq[Q]{\mathcal{H}}$}
      \NOM{H-Mix}
      \BIC{$\seq[(\piRecv{x}{}{P} \ppar Q)]{\mathcal{G} \hsep \mathcal{H} \hsep \Gamma, \tmty{x}{\bot}}$}
      \DisplayProof
    \\\\
    \AXC{$\seq[P]{\mathcal{G} \hsep \Gamma, \tmty{x}{A}}$}
    \AXC{$\seq[Q]{\mathcal{H}}$}
    \NOM{H-Mix}
    \BIC{$\seq[(P \ppar Q)]{\mathcal{G} \hsep \mathcal{H} \hsep \Gamma, \tmty{x}{A}}$}
    \SYM{\plus_1}
    \UIC{$\seq[\cpInl{x}{(P \ppar Q)}]{\mathcal{G} \hsep \mathcal{H} \hsep \Gamma, \tmty{x}{A \plus B}}$}
    \DisplayProof
    & \moveMixDown
    & \AXC{$\seq[P]{\mathcal{G} \hsep \Gamma, \tmty{x}{A}}$}
      \SYM{\plus_1}
      \UIC{$\seq[\cpInl{x}{P}]{\mathcal{G} \hsep \Gamma, \tmty{x}{A \plus B}}$}
      \AXC{$\seq[Q]{\mathcal{H}}$}
      \NOM{H-Mix}
      \BIC{$\seq[(\cpInl{x}{P} \ppar Q)]{\mathcal{G} \hsep \mathcal{H} \hsep \Gamma, \tmty{x}{A \plus B}}$}
      \DisplayProof
  \end{array}
  \]
  \caption{The disentanglement relation for \hcp.}
  \label{fig:hcp-disentangle}
\end{sidewaysfigure}
\noindent
We named this relation ``disentanglement'' to reflect the intuition that proof in \hcp represent multiple entangled \cp proofs, which we can disentangle. 

Disentanglement is terminating, and confluent up to the associativity and commutativity of mixes.
\begin{lemma}[Disentangle]\label{lem:hcp-disentangle}
  If $\seq[P]{\Gamma_1 \hsep \dots \hsep \Gamma_n}$ in \hcp, then there exist processes $\tm{P_1}, \dots, \tm{P_n}$ in \cp such that $\seq[P_1]{\Gamma_1}, \dots, \seq[P_n]{\Gamma_n}$ and
  \[
    \begin{prooftree*}
      \AXC{\seq[P]{\Gamma_1 \hsep \dots \hsep \Gamma_n}}
    \end{prooftree*}
    \moveMixDown^\star
    \begin{prooftree*}
      \AXC{$\seq[\mtf{P_1}]{\Gamma_1}$}
      \AXC{$\dots$}
      \AXC{$\seq[\mtf{P_n}]{\Gamma_n}$}
      \NOM{H-Mix$^\star$}
      \TIC{$\seq[(\mtf{P_1} \ppar \dots \ppar \mtf{P_n})]{\Gamma_1 \hsep \dots \hsep \Gamma_n}$}
    \end{prooftree*}
  \]
\end{lemma} 
\begin{proof}
  We repeatedly apply the $\moveMixDown$-rules to the derivation $\rho$ to move the mixes downwards.
  There are three cases:
  \begin{enumerate*}[label={\alph*)}]
  \item
    if a mix gets stuck above a cut, it forms a \cp cut;
  \item
    if a mix gets stuck above a $(\tens)$, it forms a \cp $(\tens)$;
  \item
    otherwise, it moves all the way to the bottom.
  \end{enumerate*}
  All applications of $(\one)$ are followed by an application of \textsc{H-Mix$_0$}, forming a \cp $(\one)$.
\end{proof}
An environment can be internalised as a type by collapsing it as a series of pars.
\begin{definition}\label{def:bigparr}
  \[
  \begin{array}{lcll}
    \ty{\bigparr(\emptycontext)}
    & = & \ty{\bot}
    \\
    \ty{\bigparr(\tmty{x_1}{A_1} , \dots , \tmty{x_n}{A_n})} 
     & = & \ty{A_1 \parr \dots \parr A_n}
     & \text{ if } n \geq 1
    \\
  \end{array}
  \]
\end{definition}\noindent
\begin{lemma}\label{lem:cp-bigparr}
  If $\seq{\Gamma}$ in \cp, then $\seq{\bigparr\Gamma}$ in \cp.
\end{lemma} 
\begin{proof}
  By repeated application of $(\parr)$.
\end{proof}\noindent
Furthermore, a hyper-environment can be internalised as a type by collapsing it as a series of tensors, where each constituent environment is internalised using $\ty{\bigparr}$. The empty hyper-environment $\ty{\emptyhypercontext}$ is internalised as the unit of tensor.
\begin{definition}\label{def:bigtens}
  \[
  \begin{array}{lcll}
    \ty{\bigtens(\emptyhypercontext)}
    & = & \ty{\one}
    \\
    \ty{\bigtens(\Gamma_1 \hsep \dots \hsep \Gamma_n)}
    & = & \ty{\bigparr\Gamma_1 \tens \dots \tens \bigparr\Gamma_n}
    & \text{ if } n \geq 1
  \end{array}
  \]
\end{definition}\noindent

\begin{theorem}\label{thm:hcp2cp-bigtens}
  If $\seq{\mathcal{G}}$ in \hcp, then $\seq{\bigtens\mathcal{G}}$ in \cp.
\end{theorem} 
\begin{proof}
  By case analysis on the structure of the hyper-environment $\ty{\mathcal{G}}$.
  If $\ty{\mathcal{G}} = \ty{\emptyhypercontext}$, we apply $(\one)$.
  If $\ty{\mathcal{G}} = \ty{\Gamma_1 \hsep \dots \hsep \Gamma_n}$, we apply \cref{lem:hcp-disentangle} to obtain proofs of $\seq{\Gamma_1}, \dots, \seq{\Gamma_n}$ in \cp, then we apply \cref{lem:cp-bigparr} to each of those proofs to obtain proofs of $\seq{\bigparr\Gamma_1}, \dots, \seq{\bigparr\Gamma_n}$, and join them using $(\tens)$ to obtain a single proof of $\seq{\bigtens\mathcal{G}}$ in \cp.
\end{proof}\noindent

\section{Related Work}
\label{sec:related-work}

Since its inception, linear logic has been described as the logic of concurrency \cite{girard1987}. Correspondences between the proof theory of linear logic and variants of the \textpi-calculus emerged soon afterwards \cite{abramsky1994,bellin1994}, by interpreting linear propositions as types for channels. Linearity inspired also the seminal theories of linear types for the \textpi-calculus \cite{kobayashi1999} and session types \cite{honda1998}. Even though the two theories do not have a direct correspondence with linear logic, the link is still strong enough that session types can be encoded into linear types \cite{dardha2017}.

It took more than ten years for a formal correspondence between linear logic and (a variant of) session types to emerge, with the seminal paper by Caires and Pfenning~\cite{caires2010}. This inspired the development of Classical Processes by Wadler~\cite{wadler2012}.

The idea of using hypersequents to capture parallelism in linear logic judgements is not novel: Carbone \etal~\cite{carbone2018} extended the multiplicative-additive fragment of intuitionistic linear logic with hypersequents to type global descriptions of process communications known as choreographies. This work is distinct from our approach in that \hcp is based on classical linear logic and manipulates hypersequents differently: in Carbone \etal~\cite{carbone2018}, hypersequents can be formed only when sequents share resources (\cf \textsc{H-Mix}), and resource sharing is then tracked using an additional connection modality (which is not present in \hcp).

\section{Conclusions and Future Work}
\label{sec:conclusion}
In this paper, we introduced \hcp, a variant of \dhcp which sits in between \cp and \dhcp. It has reduction semantics, and does not allow for delayed actions, like \cp, but registers parallelism using hypersequents, like \dhcp. This results in a calculus which structurally resembles \dhcp, but which is behaviourally much more like \cp: all \cp processes can be translated to \hcp processes, and this translation preserves the reduction behaviour of the process. The key insight to making this calculus work is to add a side condition to the logical rules in the type system which rules out self-locking processes---processes which act on both endpoints of a channel, \eg $\tm{\piRecv{x}{}{\piBoundSend{x}{}{\piHalt}}}$.

\hcp focuses on the basic features of \cp, corresponding to multiplicative applicative linear logic. In the future, we intend to study the full version of \hcp, which includes exponentials and quantifiers. Furthermore, we intend to study extensions to \hcp which capture more behaviours, such as recursive types~\cite{lindley2016} and access points~\cite{gay2009}. Separately, we would like to extend the reduction semantics of \hcp to cover all of \dhcp by adding delayed actions.

\bibliographystyle{eptcs}
\bibliography{main}
\end{document}